\documentclass[12pt]{amsart}

\usepackage{amssymb,amsfonts, color,epsf}
\usepackage [cmtip,arrow]{xy}
\xyoption{all}
\usepackage {pb-diagram,pb-xy}
\usepackage{enumerate}
\usepackage{accents}
\usepackage[noautoscale]{youngtab}
\usepackage{subfigure}

\ifx\pdfpageheight\undefined
\usepackage[dvips,colorlinks=true,linkcolor=blue,citecolor=red,%
urlcolor=green]{hyperref}
\usepackage[dvips]{graphicx}
\makeatletter
\edef\Gin@extensions{\Gin@extensions,.mps}
\makeatother
\else
\usepackage[pdftex]{graphicx}
\usepackage[bookmarksopen=false,pdftex=true,breaklinks=true,%
backref=page,pagebackref=true,plainpages=false,%
hyperindex=true,pdfstartview=FitH,colorlinks=true,%
pdfpagelabels=true,colorlinks=true,linkcolor=blue,%
citecolor=red, urlcolor=darkgray, hypertexnames=false%
]%
{hyperref}
\fi
\usepackage{bm}

\usepackage{amssymb,color,epsf}
\usepackage{mathtools}
\usepackage{enumerate}
\usepackage{float}

\usepackage{amsmath}
\usepackage{tabularx}
\usepackage{color, colortbl}
\usepackage{url}

\DeclareMathAlphabet{\mathpzc}{OT1}{pzc}{m}{it}
\usepackage[margin=1.5in]{geometry}

\usepackage {pb-diagram,pb-xy}
\usepackage[mathscr]{eucal}
\usepackage[utf8]{inputenc}
\usepackage[T1]{fontenc}
\usepackage[english]{babel}
\usepackage{bm}
\usepackage{tabularx}
\usepackage{ltablex}

\usepackage{bm}

\usepackage{amsmath,amssymb,amsfonts}
\newtheorem{theorem}{Theorem}
\newtheorem{lemma}{Lemma}[section]
\newtheorem{corollary}{Corollary}
\newtheorem{proposition}{Proposition}[section]

\newtheorem*{claim*}{Claim}

\newtheorem*{theorem*}{Theorem}
\newtheorem*{corollary*}{Corollary}
\theoremstyle{definition}
\newtheorem{definition}{Definition}[section]

\newtheorem{notation}{Notation}[section]

\theoremstyle{remark}
\newtheorem{remark}{Remark}

\definecolor{DarkBlue}{rgb}{0,0.1,0.55}

\numberwithin{equation}{section}

\newcommand {\hide}[1]{}
\newcommand {\sign} {\mbox{\bf sign}}

\newcommand {\junk}[1]{}

\newcommand {\R} {\mathrm{R}}

\newcommand {\C}     {\mathbb{C}}

\newcommand {\Real}{\mathrm{Re}}
\newcommand {\Imaginary}{\mathrm{Im}}

\newcommand {\la}   {{\langle}}
\newcommand {\ra}   {{\rangle}}
\newcommand {\eps} {{\varepsilon}}

\newcommand {\Sign}      {\mbox{\rm Sign}}

\newcommand{\card}{\mathrm{card}}

\def\addots{\mathinner{\mkern1mu
		\raise1pt\vbox{\kern7pt\hbox{.}}
		\mkern2mu\raise4pt\hbox{.}\mkern2mu
		\raise7pt\hbox{.}\mkern1mu}}

\usepackage{algorithm}
\usepackage{algpseudocode}
\algnewcommand\algorithmicinput{\textbf{Input:}}
\algnewcommand\INPUT{\item[\algorithmicinput]}
\algnewcommand\algorithmicoutput{\textbf{Output:}}
\algnewcommand\OUTPUT{\item[\algorithmicoutput]}

\algnewcommand\algorithmicproc{\textbf{Procedure:}}
\algnewcommand\PROCEDURE{\item[\algorithmicproc]}
\algnewcommand\algorithmiccomplexity{\textbf{Complexity:}}
\algnewcommand\COMPLEXITY{\item[\algorithmiccomplexity]}
\usepackage{multicol}
\usepackage{tikz-cd}
\usepackage{tabu}
\usepackage{caption}
\usepackage{booktabs}
\usepackage{makecell}
\newcolumntype{P}[1]{>{\centering\arraybackslash}p{#1}}

\makeatletter 
\renewcommand\p@enumii{}
\makeatother
\title{Quantum Analog of Shannon's Lower Bound 
Theorem}
\author{Saugata Basu}
\address{Department of Mathematics,
Purdue University, West Lafayette, IN 47906, U.S.A.}
\email{sbasu@math.purdue.edu}
\urladdr{www.math.purdue.edu/~sbasu}
\author{Laxmi Parida}
\address{
IBM T.J. Watson Research Center,
Yorktown Heights, NY 10598.
}
\email{parida@us.ibm.com}
\urladdr{researcher.ibm.com/person/us-parida}

\begin{document}
\begin{abstract}
    Shannon proved that almost all Boolean functions require a circuit of size $\Theta(2^n/n)$ \cite{Shannon}. We prove a quantum analog of this classical result. Unlike in the classical case the number of quantum circuits of any fixed size that we allow  is uncountably infinite. Our main tool is a classical result in real algebraic geometry bounding the number of realizable sign conditions of any finite set of real polynomials in many variables.
\end{abstract}	
\maketitle

 \section{Introduction}
Boolean circuit complexity measures complexity of Boolean functions $f: \{0,1\}^n \rightarrow \{0,1\}$ by the size of the smallest circuit computing  $f$.
We assume that the gates of a circuit come from some fixed finite universal set of gates.
We refer the reader to \cite[Page 7, Definition 3.1]{Wegener} for the precise definition of a circuit.
Notice that in this case the set of circuits of bounded size is finite. 

The study of circuit complexity of Boolean functions naturally leads to the definitions of ``non-uniform'' analogs of classical complexity classes such as $\mathbf{P}$. For example, the non-uniform analog of the class $\mathbf{P}$ is the set of sequences of Boolean functions
\[
\left(f: \{0,1\}^n \rightarrow \{0,1\}\right)_{n >0}
\]
for which there exists a polynomial $p(n)$ such that for each
$n>0$, there exists a circuit of size bounded by $p(n)$ computing $f_n$. 

Shannon \cite{Shannon} showed using a counting argument (using the fact that as noted before that the set of circuits of bounded size is finite)  very early that almost all Boolean functions need circuits of size $\Omega({2^n}/{n})$ (see \cite[Page 90, Theorem 2.1]{Wegener}.
Here almost all refers to the fact that the number of Boolean functions that require circuits of size $\Omega({2^n}/{n})$ is bounded from below by $2^{2^n}(1 - o(1))$. It is also known (see for instance \cite[Page 92, Theorem 2.2]{Wegener}) that every Boolean
function $f: \{0,1\}^n \rightarrow \{0,1\}$ can be computed by a circuit of size $O({2^n}/{n})$.

\begin{remark}
\label{rem:arity-classical}
In the theorems cited above the set of gates used is assumed to be all possible gates 
of arity two (i.e. the finite set of $2^{2^2} = 16$ gates computing all possible Boolean functions 
$g:\{0,1\}^2 \rightarrow \{0,1\}$). The arity two is not important. The same asymptotic results will
hold even if we allow all gates of some bounded arity $q \geq 2$ 
(the constants will depend on the arity) .
\end{remark}

In this paper we are are concerned with proving analogous results for quantum circuits.

\hide{
In order to
define circuits one should first fix a set of gates. We assume each gate in the set has arity 
(also commonly  referred to as fan-in) bounded by a fixed constant. 
For example, for classical circuits in 
the set-up for Shannon's theorem one assumes that set of gates is a finite complete set (for example,
the AND and the NOT gates). Notice that in this case the set of circuits of bounded size is finite. 
}

Quantum complexity theory is a relatively new discipline. A quantum circuit computes a unitary transformation of a certain finite but exponentially large dimensional Hilbert space. 
We refer the reader to \cite[Definition 6.1]{Kitaev} for the precise definition of a quantum circuit.
The Hilbert space
comes equipped with a computational basis whose elements should be thought of as the elements of the
Boolean hypercube. There is a standard definition of what it means for such a circuit to compute a Boolean function $f:\{0,1\}^n \rightarrow \{0,1\}$ which we explain later in the paper. The notion of a 
size of a quantum circuit mirrors the classical definition (i.e. the number of gates) though we also take into account the number of ancillary qubits which form the workspace of the quantum circuit. But there is one crucial difference between classical and quantum circuits which we explain below.

In analogy with classical circuits (cf. Remark~\ref{rem:arity-classical}),
by a quantum circuit we will mean a circuit as defined in \cite[Definition 6.1]{Kitaev} 
where we allow all quantum gates of some bounded arity. 
Even if one fixes the arity (say $q > 0$) of a quantum gate, unlike in the classical case there is a \emph{continuum many}  choices of such a gate -- namely, each element of the 
unitary group 
\footnote{The unitary group $\mathbf{U}(N)$ is the group of $N\times N$ complex matrices $U$
satisfying $U U^{\dagger} = \mathrm{Id}_n$. It is a real Lie group of dimension $N^2$.}
$\mathbf{U}(2^q)$ is a possible quantum gate of arity $q$. The real dimension of the 
group $\mathbf{U}(2^q)$ is equal to $2^{2q}$. Thus, with such an uncountable choice of the set of gates,
the cardinality of the set of quantum circuits of size bounded by any fixed positive number is 
also uncountable. Thus it makes sense to ask whether using this additional flexibility quantum circuits of  ``small'' (in terms of $n$)  sizes can actually compute all the $2^{2^n}$ Boolean functions $f:\{0,1\}^n \rightarrow \{0,1\}$. The main result of the paper (see Corollary~\ref{cor:main}) can be colloquially framed as saying that Shannon's lower bound for classical circuits also hold
for quantum circuits (even after we allow arbitrary quantum gates of bounded arity) i.e. almost all Boolean functions $f:\{0,1\}^n \rightarrow \{0,1\}$ require quantum circuits of size $\Omega(2^n/n)$.
Since a classical circuit computing a Boolean function can be converted into a quantum circuit with at most a constant factor increase in size it follows that every Boolean function $f:\{0,1\}^n \rightarrow \{0,1\}$ can be computed by a quantum circuit of size at most $O(2^n/n)$. Thus, it is the lower bound
part of Shannon's result that is of main interest.

We prove our quantum version of Shannon's theorem by first proving an upper bound on the number
of Boolean functions that can be computed by a quantum circuit of a given size (see Theorem~\ref{thm:main} below).
A key tool in the proof of Theorem~\ref{thm:main} is a bound from real algebraic geometry (that we explain in Section~\ref{subsec:rag}) which gives an upper bound on the number of realizable sign conditions of any finite  set of real polynomials in several indeterminates of bounded degrees. This result actually bounds the zero-th Betti number of the realizations of all sign conditions of the set of polynomials and has generalizations to higher Betti numbers as well \cite{BPR02}. It has previously being used 
in proving upper bounds on the number of configurations in various geometric settings (see for example
\cite{GP2} for one such example), but to the best of our knowledge has not being used in proving lower 
bounds in quantum complexity theory.

The rest of the paper is organized as follows. In Section~\ref{sec:main} we state our main results.
In Section~\ref{sec:proof} we prove the main theorem and its corollary. In Section~\ref{sec:alternative}
we discuss an alternative approach towards proving the main result of this using variants of Solovay-Kitaev approximation and explain its deficiency. In Section~\ref{sec:conclusion} we discuss some 
possible future work.

\section{Main Results}
\label{sec:main}
We follow the usual conventions. 
We denote by $|0\ra, |1\ra$ the computational basis for each qubit.
We identify a $0$-$1$ string $x_{n-1}\cdots x_{0}$ of length $n$, with the integer
\[
x = \sum_{k=0}^{n-1} x_k \cdot 2^k,
\]
and denote the corresponding (separable) state
$|x_{n-1}\ra\otimes \cdots \otimes |x_0\ra$ simply by $|\mathbf{x}\ra$.
Similarly, we will often abbreviate 
$|0\ra \otimes \cdots  \otimes |0\ra$ by $|\mathbf{0}\ra$.

For $q \geq 1$ we denote by $\mathcal{U}_q$ the (uncountably infinite) set of quantum gates with at most  $q$ inputs.
Each gate $g \in \mathcal{U}_q$ corresponds to a unitary transformation $U_g \in \mathbf{U}(2^{k_g})$,
where $k_g \leq q$ is the arity (fan-in) of the gate $g$.
For $q \geq 2$, $\mathcal{U}_q$ is a universal family. Namely, any unitary transformation can be implemented by a quantum circuit with gates from $\mathcal{U}_2$.

A quantum circuit $C \in \mathcal{C}_1$ with $n$ input qubits, 
(whose values will be denoted by $x_1,\ldots,x_n$) and $t$ ancillary qubits
(whose values will be denoted by $z_1,\ldots,z_t$),
and  having $r$ gates drawn from $\mathcal{U}_q$ is determined by the following data:
\begin{enumerate}[1.]
    \item an ordering of the $r$ gates 
(lets suppose the ordered tuple of gates is $g_1,\ldots,g_r$, with the gate $g_i$ having arity $k_i \leq q$), and
\item 
for each $i, 1 \leq i \leq r$,  an ordered  choice of $k_i$ elements
from amongst $x_1,\ldots,x_n,z_1, \ldots z_t$.
\end{enumerate}

For $q \geq 1$, 
we will denote by $\mathcal{C}_q$, the set of quantum circuits 
using gates from the set $\mathcal{U}_q$.


For $C$, a quantum circuit with $n$ qubits as input and $t$ ancillary qubits,
we denote by $U(C) \in \mathbf{U}(N)$ the unitary transformation implemented by $C$ where $N = 2^{n+t}$. 

\hide{
Following \cite{},
for a Boolean function $f: \{0,1\}^n \rightarrow \{0,1\}$, we will denote by $U_f \in \mathbf{U}(n+1,\C)$ the unitary
transformation that for each $j, 0 \leq j \leq N-1$, and $y \in \{0,1\}$
\[
U_f(j\ra \otimes |y \ra)  =   |j \ra \otimes |f(j)\oplus y \ra.
\]
}

\hide{
\begin{definition}[Exact computation of Boolean functions by a quantum circuit]
Let $f: \{0,1\}^n \rightarrow \{0,1\}$ denote a Boolean function on $n$ Boolean bvariables.
We say that a quantum circuit $C$ on $n$-qubits using $t+1$ ancillary bits 
$y,z_1,\ldots,z_t$
\emph{computes $f$ exactly} if 
for each $j, 0 \leq j \leq 2^{n}-1$, $y \in \{0,1\}$, and some fixed 
$z_0, 0 \leq z_0 \leq 2^{t} -1$
\[
U(C)(|j \ra \otimes |y\ \ra \otimes 
| z_0 \ra
= |j \ra \otimes |f(j)\oplus y \ra
\otimes 
| z_0 \ra
\]
\end{definition}
}

We now explain what is meant by a quantum circuit computing a Boolean function. 
For ease of understanding,  we start with a provisional (very stringent) definition 
and then give a more general (much less stringent)  definition afterwards. We will use the more general (less stringent) definition in the rest of the paper noting that a lower bound result is more powerful if the definition is less stringent.

Let $f: \{0,1\}^n \rightarrow \{0,1\}$ be a Boolean function. We associate a 
unitary transformation $U_f \in \mathbf{U}(2^{n+1})$ to $f$ which takes for
$x_0,\ldots,x_{n-1},y \in \{0,1\}$, the separable state 
$|\mathbf{x} \ra \otimes |y\ra$ to the state 
$|\mathbf{x} \ra \otimes |f(x) \oplus y\ra$, where $\oplus$  denotes ``exclusive-or''.
This motivates calling 
a quantum circuit $C$ taking as input $n$ qubits $x_0,\ldots,x_{n-1}$ and an ancillary quibit
$y$, such that $U(C) = U_f$, a circuit computing \emph{stringently} the Boolean function $f$. Note that using the measurement postulate of quantum mechanics, if we 
set the input qubits to $|\mathbf{x}\ra \otimes |0\ra$ and measure
the ancillary bit in the output, we will be left in the state $|1\ra$ if $f(x) = 1$ and 
in the state $|0\ra$ if $f(x) = 0$ \emph{with probability $1$}. 

We can state our main result already for the stringent model described above. Later we will state and prove 
a more powerful result by making the definition less stringent 
to take into account a much looser notion of
an acceptable output for the quantum circuit computing a Boolean function and also allow additional ancillary bits (exponentially many).
The following theorem can be deduced from Corollary~\ref{cor:main} stated later.

For a quantum circuit $C$, we will denote by $\mathcal{G}(C)$ the set of gates of $C$.
\begin{theorem*}
For each $q \geq 1$, there exists $c=c_q, 0 < c <1$ (depending only on $q$), 
     such that for each $n >0$ 
    the number of distinct Boolean functions on $n$ variables
    that can be computed stringently by quantum circuits belonging to $\mathcal{C}_q$,  
    with 
    \[
    \card(\mathcal{G}(C)) \leq c \cdot \frac{2^n}{n}
    \]
    is bounded by 
    \[
    2^{2^{n-1}}.
    \]
    Consequently, the fraction of Boolean functions that need quantum circuits in $\mathcal{C}_q$ of size greater than $c \cdot \frac{2^n}{n}$ is $1- 2^{-2^{n-1}} = 1 - o(1)$.  
\end{theorem*}

\begin{remark}
Notice that since the set of gates $\mathcal{U}_q$ is infinite (uncountably so), and hence the number of 
distinct quantum circuits of any bounded size is also uncountably infinite. So a priori there is no reason for such circuits not being able to compute all Boolean functions on $n$ variables.
\end{remark}

\begin{remark}
    A lower bound similar to the one in the theorem stated above with a finite choice of the gate set has appeared in the literature \cite[Claim F.1]{chia2021quantum}. The proof of the result in \cite{chia2021quantum} strongly uses the finiteness of the set of gates and indeed the constant in the theorem depends on the cardinality of this set. The import of our result is that the choice of quantum gates we allow is uncountable (only the arity is fixed) -- while the lower bound stays the same.
\end{remark}

As mentioned before we will work with a more general notion of what it means for a 
quantum circuit to compute a Boolean function. We relax our prior definition in two
ways. First, instead of ending up in the right state depending on the value of the function $f$ 
with probability $1$, we are satisfied if we end at the right state with probability $> 1/2$.
We will also allow more than one ancillary bits. The precise  definition is as follows.

\begin{definition}[Computation of Boolean functions by a quantum circuit]
\label{def:quantum-boolean}
A quantum circuit $C$ on $n$ qubits denoted $x_0, \ldots,x_{n-1}$
and $t+1$ ancillary bits 
denoted by $y,z_1,\ldots,z_t$,
 \emph{computes $f$}, 
if 
for each $x,x', 0 \leq x,x' \leq 2^{n}-1$, $y \in \{0,1\}$, 

\hide{
\begin{eqnarray*}
|a_{j,j',y,0}|^2 &> & 1/2 \mbox{ if $j' = j$, $y = 0$ and $f(j) = 1$ or $j' = j, y =1$ and $f(j) = 0$}, \\
|a_{j,j',y,0}|^2 &< & 1/2 \mbox{ otherwise},
\end{eqnarray*}
}
\begin{eqnarray}
\label{eqn:def:quantum-boolean}
\nonumber
|\la \mathbf{x}' \;y \; \mathbf{0}  \; |\; U(C) \;|\; \mathbf{x} \; 0 \; \mathbf{0} \ra|^2 &>& 1/2 \text{ if $x' = x$ and $y\oplus f(x_0,\ldots,x_{n-1}) = 0$},\\  
&<&  1/2 \text{ otherwise.}
\end{eqnarray}

\hide{
$a_{j,j',y,z}$ is defined by the equation,
\[
U(C)(|j \ra \otimes |y\ra \otimes 
| z \ra
=
\sum_{0 \leq j' \leq 2^{n}-1, 0 \leq z \leq 2^{t}-1} a_{j,j',y,z} |j'\ra \otimes |y\ra  \otimes |0\ra \otimes |z\ra.
\]
}
($|\mathbf{x} \ra \otimes |y \ra \otimes |\mathbf{0}\ra$ is abbreviated as $|\mathbf{x} \; y\; \mathbf{0}\ra$).
\end{definition}

\begin{remark}
    \label{rem:approximate-implies-exact}
    Note that if a quantum circuit computes a Boolean function $f$ stringently, then it also computes $f$ according to Definition~\ref{def:quantum-boolean}. Also note that the inequality appearing
    in \eqref{eqn:def:quantum-boolean} is the most relaxed possible in as much as we do not insist
    on any positive gap between the accepting and rejecting probabilities. This point will be important 
    later when we discuss an alternative possible method for obtaining lower bounds using the Solovay-Kitaev approximation theorem (see Section~\ref{sec:alternative}).
\end{remark}

We now fix a notion for size of quantum circuits.

\begin{definition}
    For $C$ a quantum circuit taking as input $n$ qubits. We denote by $t(C)$ the number of anicllary qubits used by $C$ and $r(C)$ the number of gates in $C$. We will denote by $s(C) = t(C) + r(C)$
    and call $s(C)$ the \emph{size} of $C$. We will denote the set of all quantum circuits $C \in \mathcal{C}_q$ with $n$ inputs, and with $t(C) = t$ and $r(t) =  r$, by
    $\mathcal{C}_{q,n,r,t}$.
\end{definition}

Our main theorem is the following.

\begin{theorem}
\label{thm:main}
     The 
     number of distinct Boolean functions $f: \{0,1\}^n \rightarrow \{0,1\}$ which can be computed approximately by a quantum circuit belonging to $\mathcal{C}_{q,n,r,1+t}$ 
     is bounded by
    \[
    t^{q \cdot r} \cdot (2^{n} \cdot  r)^{c \cdot 2^{2q} \cdot r} 
 \]
    for some universal constant $c > 0$.
\end{theorem}

Theorem~\ref{thm:main} has the following important corollary.
\begin{corollary}
\label{cor:main}
    There exists $c=c_q, 0 < c <1$ (depending only on $q$), such that for each $n >0$, 
    the number of distinct Boolean functions on $n$ variables
    that can be computed  by quantum circuits belonging to 
    \[
    \mathcal{C}_{q,n, c \cdot \frac{2^n}{n}, 2^n}
    \]
    is bounded by 
    \[
    2^{2^{n-1}}.
    \]
    Consequently, the fraction of Boolean functions that need quantum circuits of size  
    $\Omega(\frac{2^n}{n})$ is $1- 2^{-2^{n-1}} = 1 - o(1)$.  
\end{corollary}

\hide{
We will denote by 
\begin{eqnarray*}
|\psi(C)\ra &=& \sum_{j=0}^{N-1} a_j(C) |j\ra \\
&=& U(C)(|0\ra).  
\end{eqnarray*}
}

\section{Proof of Theorem~\ref{thm:main} and Corollary~\ref{cor:main}}
\label{sec:proof}
For $g \in \mathcal{G}(C)$, we denote by $U_{g,C}$ the $2^{k_g} \times 2^{k_g}$
unitary matrix corresponding to $g$,  where $k_g$ is the arity of the gate $g$.

Let $C \in \mathcal{C}_{q,n,r,t}$ and denote the sequence of gates of $C$ 
    by $g_1,\ldots,g_r$ with corresponding  arities $k_1,\ldots,k_r$.
    For each $i, 1 \leq i \leq r$, denote the entries of the two 
    real $2^{k_i} \times 2^{k_i}$ matrices,
    $\Real(U_{g_i,C}), \Imaginary(U_{g_i,C})$ by 
    $\left(v^{(i)}_{h,h'}\right)_{1 \leq h,h'\leq 2^{k_i}}$ and
    $\left(w^{(i)}_{h,h'}\right)_{1 \leq h,h' \leq 2^{k_i}}$.
    So,
    \[
    U_{g_i,C} = \left(v^{(i)}_{h,h'}\right)_{1 \leq h,k \leq 2^{k_i}} + \sqrt{-1}\cdot \left(w^{(i)}_{h,h'}\right)_{1 \leq h,h' \leq 2^{k_i}}.
    \]

    We use the notation introduced above in the following lemma. We use the convention that
    upper case letters with indices such as $V^{(i)}_{h,h'}, W^{(i)}_{h,h'}$ are used to denote indeterminates in certain polynomials and the corresponding lower case letters $v^{(i)}_{h,h'}, w^{(i)}_{h,h'}$ are used to denote real numbers to which the corresponding indeterminates 
    are specialized.
    
\begin{lemma}
\label{lem:coefficients}
    For each $x,x',  0 \leq x,x' \leq 2^{n}-1, z,z', 0 \leq z,z' \leq 2^{t}-1$, there exists a real polynomial
    \[
    P_{| \mathbf{x}\mathbf{z}\ra,|\mathbf{x}'\mathbf{z}'\ra,C} \in \R\left[\left(V^{(i)}_{h,h'}, W^{(i)}_{h,h'}\right)_{\substack{1 \leq i \leq r,\\ 1 \leq h,h' \leq 2^{k_i}}}\right], 
    \]
    with $\deg(P_{| \mathbf{x}\mathbf{z}\ra,|\mathbf{x}'\mathbf{z}'\ra,C}) \leq  2 r$, 
    such that
    \[
        |\la \mathbf{x} \mathbf{z}\;|\;U(C) \;|\; \mathbf{x}'\mathbf{z}' \ra|^2 = 
        P_{| \mathbf{x}\mathbf{z}\ra,|\mathbf{x}'\mathbf{z}'\ra,C}\left((v^{(i)}_{h,h'},w^{(i)}_{h,h'})_{\substack{1 \leq i \leq r,\\ 1 \leq h,h' \leq 2^{k_i}}}\right).
    \]
\end{lemma}

\begin{proof}
Let $N = 2^{n+t}$.
    Observe that $\la \mathbf{x} \mathbf{z}\;|\; U(C) \;|\; \mathbf{x}' \mathbf{z}'\ra$ is the 
    $(2^t\cdot x+z,2^t\cdot x'+z')$-th entry of the 
$N\times N$ unitary matrix $U(C)$ where $U(C)$ is a product of 
$r$ many 
$N \times N$ unitary matrices 
of the form $U_{g_i,C} \otimes \mathbf{1}_{N - 2^{k_i}}$.

Observe that the $(h,h')$-th entry of $U_{g_i,C}$ equals $v^{(i)}_{h,h'} + \sqrt{-1}\cdot w^{(i)}_{h,h'}$. 
Thus, the real and the imaginary parts of each entry of $U(C)$ is a polynomial in the real numbers
$v^{(i)}_{h,h'} , w^{(i)}_{h,h'}, 1 \leq i \leq r, 1 \leq h,h' \leq 2^{k_i}$
having degree at most $r$, and
taking the square of the modulus gives a polynomial of degree at most $2 r$.
\end{proof}

Next, we count the number of distinct ``topologies'' underlying quantum circuits of bounded size.
In order to make this precise we introduce the following definition.

\begin{definition}
\label{def:topology}
    For any quantum circuit $C$, we will denote by $\widetilde{C}$ the quantum circuit obtained by replacing for each $g \in \mathcal{G}(C)$, the gate $g$, by $g_0$ with the same input and output  but with $U_{g_0} = \mathbf{1}_{2^{k_g}}$ where $k_g$ is the arity of $g$.
    
    For any two quantum circuits 
    $C,D \in \mathcal{C}_{q,n,r,t}$, we will say that
    $C,D$ are equivalent (denoted $C \sim_{q,n,r,t} D$) if $\widetilde{C} = \widetilde{D}$. 
    Clearly $\sim_{q,n,r,t}$
    is an equivalence relation.

    We denote by 
    $\mathcal{T}_{q,n,r,t}$ the set of equivalence classes of $\sim_{q,n,r,t}$. 
\end{definition}

\begin{remark}
    The set $\mathcal{T}_{q,n,r,t}$ should be thought of as the set of underlying ``topologies''
    of quantum circuits with $r$ gates with fan-ins at most $q$,
    with $n$ qubits as input and $t$ ancillary bits. 
\end{remark}

The following lemma establishes an upper bound on the cardinality of $\mathcal{T}_{q,n,r,t}$ in terms of
$q,n,r$ and $t$.

\begin{lemma}
\label{lem:topology}
    For  every $q, n,r,t > 0$,  
    \[
    \card(\mathcal{T}_{q,n,r,t}) \leq  q^r \cdot (n+t)^{q\cdot r}.
    \]
\end{lemma}

\begin{proof}
Each of the $r$ gates have arities $\leq q$. The number of $r$-tuples $(k_1,\ldots,k_r)$ of possible arities, where $1 \leq k_i \leq q$ is equal to $q^r$. The equivalence class $T$ of a circuit is determined by the 
ordered choice of  $k_i \leq q$ from amongst the input qubits $x_1,\ldots,x_n$ and the $t$ ancillary bits $z_1,\ldots,z_t$. The number of choices for the 
$i$-th gate $g_i$ is 
\[
(k_i)! \binom{n+t}{k_i} \leq (n+t)^{k_i} \leq (n+t)^q.
\]
Since there are $r$ gates
the total number of choices is bounded by 
\[
q^r \cdot (n+t)^{q \cdot r}.
\]
\end{proof}

\begin{remark}
    \label{rem:crude}
    Lemma~\ref{lem:topology} gives a rather crude bound that suffices for our purpose. It is possible to prove a much tighter upper bound (see for example \cite[Page 88, Lemma 2.1]{Wegener}).
\end{remark}

We will now fix an equivalence class $T \in \mathcal{T}_{q,n,r, 1+t}$ and ask how many distinct Boolean 
functions can be computed  (cf. Definition~\ref{def:quantum-boolean}) using quantum circuits belonging to $T$. Our strategy is to identify a finite set of real polynomials (denoted by
$\mathcal{P}_T$ below), and associate to each Boolean function $f:\{0,1\}^n \rightarrow \{0,1\}$ which
is computed  by some circuit $C \in T$ a unique realizable sign condition (see below). Our bound on the number of Boolean functions that can be computed
using circuits in $T$ will then follow from an upper bound on the number of realizable sign conditions 
on $\mathcal{P}_T$ which is furnished by a classical result in real algebraic geometry.

\subsection{Bound on the number of realizable sign condition}
\label{subsec:rag}
Suppose that $\mathcal{P}$ be a finite set of polynomials in $\mathbb{R}[X_1,\ldots,X_k]$.
For each $x \in \mathbb{R}^k$, and $P \in \mathbb{R}[X_1,\ldots,X_k]$, we define 
\[
\sign(P(x)) = 
\begin{cases}
    +1, \text{if $P(x) > 0$}, \\
    -1, \text{if $P(x) < 0$}, \\
    0,  \;\text{if $P(x) = 0$}.
\end{cases}
\]

Similarly, for a finite tuple of polynomials $\mathcal{P}$ and $x \in \mathbb{R}^k$, we similarly define
\[
\sign(\mathcal{P}(x)) = (\sign(P))_{P \in \mathcal{P}} \in \{0,1,-1\}^{\mathcal{P}}.
\]

We say that an element $\sigma \in \{0,1,-1\}^{\mathcal{P}}$ is a \emph{realizable sign condition of 
$\mathcal{P}$} if there exists $x \in \mathbb{R}^k$ such that 
\[
\sign(\mathcal{P}(x)) = \sigma.
\]

We denote the set of realizable sign conditions of a finite tuple $\mathcal{P}$ of polynomials in
$\mathbb{R}[X_1,\ldots,X_k]$ by $\Sign(\mathcal{P}) \subset \{0,1,-1\}^{\mathcal{P}}$.

Observe that if $\mathcal{P}$ has length $N$, then the cardinality of $\Sign(\mathcal{P})$ could potentially be as large as $3^N$. It is an important result in real algebraic geometry (with many applications), that if $k$ as well as the degrees of the polynomials in $\mathcal{P}$ are small compared to
$N$, then the $\card(\Sign(\mathcal{P}))$ is much smaller than $3^N$. The precise result 
that we need is the following.

\hide{
We will also need the following key result bounding the number of realizable sign conditions of a set of real polynomials. We first introduce a notation.


\begin{notation}[Sign conditions and their realizations]
\label{not:sign-condition}
Let $\mathcal{P}$ be a finite subset of $\mathbb{R}[X_1,\ldots,X_k]$.
For $\sigma \in \{0,1,-1\}^{\mathcal{P}}$, and $x \in \mathbb{R}^k$ we denote by
$\sign(\mathcal{P})(x) = (\sign(P(x)))_{P \in \mathcal{P}}$.

For $\sigma \in \{0,1,-1\}^{\mathcal{P}}$, we call the formula $\bigwedge_{P \in \mathcal{P}} (\sign(P) = \sigma(P))$
to be a 
\emph{sign condition on $\mathcal{P}$} and call its realization the \emph{realization of the sign condition $\sigma$}.
We denote the set of all realizable sign condition on $\mathcal{P}$ by $\Sign(\mathcal{P})$.
\end{notation}

We will need the following key proposition.
}

\begin{proposition} \cite[Proposition 7.31]{BPRbook2}
\label{prop:sign-condition}
Let $\mathcal{P}$ be a finite tuple of polynomials in $\mathbb{R}[X_1,\ldots,X_k]$.
Then, 
\[
\card(\Sign(\mathcal{P})) \leq \sum_{1 \leq j \leq k} \binom{N}{j} 4^j d(2d -1)^{k-1} = (O(N d))^k,
\]
where $N = \mathrm{length}(\mathcal{P})$ and $d = \max_{P \in \mathcal{P}} \deg(P)$.
\end{proposition}

\begin{remark}
Various versions of Proposition~\ref{prop:sign-condition} are known (see for example \cite{Warren,Alon}).
It appears in the more precise form as stated above in \cite{BPR02}, where the bound is proved on the 
number of connected components of the realizable sign conditions (i.e. the sum of the zero-th Betti number of the realizations of each sign condition).
Note this is a priori larger than just the number of realizable sign conditions.
\end{remark}

We now return to the proof of Theorem~\ref{thm:main} by proving an upper bound on the number 
of distinct Boolean functions that can be computed by quantum circuits belonging to a fixed equivalence class $T \in \mathcal{T}_{q,n,r,1+t}$.

First notice that for each  $T \in \mathcal{T}_{q,n,r,1+t}$ and $C,C' \in T$, 
$0 \leq x,x' \leq 2^{n}-1$, $y,y' \in \{0,1\}$ and $0 \leq z,z' \leq 2^{t}-1$, 
    \[
    P_{| \mathbf{x}y\mathbf{z}\ra,|\mathbf{x}'y'\mathbf{z}'\ra,C} =
    P_{| \mathbf{x}y\mathbf{z}\ra,|\mathbf{x}'y'\mathbf{z}'\ra,C'} 
    \]
(see Lemma~\ref{lem:coefficients}).
We will denote  by $P_{| \mathbf{x}y\mathbf{z}\ra,|\mathbf{x}'y'\mathbf{z}'\ra,T}$ the polynomial
$P_{| \mathbf{x}y\mathbf{z}\ra,|\mathbf{x}'y'\mathbf{z}'\ra,C}$ for some and hence all $C \in T$.

\begin{lemma}
\label{lem:0-1}
    For each $T \in \mathcal{T}_{q,n,r, 1+t}$, 
    the number of distinct
    Boolean functions computed by some quantum circuit $C \in T$ 
    is bounded by 
    \[
    (2^{n} \cdot r)^{O(2^{2q} \cdot r)}. 
    \]
\end{lemma} 

\begin{proof}
Let $\mathcal{P}_{T}$  denote the tuple of polynomials
\[
\left(
P_{|\psi\ra,|\psi'\ra,T} - 1/2
\right)_
{
\substack{|\psi\ra = |\mathbf{x} \; 0 \; \mathbf{0}\ra, 
|\psi'\ra = |\mathbf{x}' \; y \; \mathbf{0} \ra \\ 
0 \leq x,x' \leq 2^{n}-1,\\
y \in \{0,1\}
}
}.
 \]

Suppose that $f: \{0,1\}^n \rightarrow \{0,1\}$ is a Boolean function computed  
by $C \in T$. We will follow the notation introduced in Lemma~\ref{lem:coefficients} and denote by

\[
\left(v^{(i)}_{h,h'}, w^{(i)}_{h,h'}\right)_{\substack{1 \leq i \leq r,\\ 1 \leq h,h' \leq 2^{k_i}}}
\]
the tuple
of real numbers which correspond to the real and imaginary parts of the unitary matrices corresponding to the gates of $C$. 

Then, for $0 \leq x,x' \leq 2^{n}-1$,
$|\psi\ra =  |\mathbf{x} \; 0 \;  \mathbf{0}\ra$,
$|\psi'\ra = |\mathbf{x}' \; y \; \mathbf{0} \ra$,
the sign of the polynomial 
\[
P_{|\psi \ra,|\psi'\ra,T}
\]
evaluated at the point
\[
\left(v^{(i)}_{h,h'}, w^{(i)}_{h,h'}\right)_{\substack{1 \leq i \leq r, \\ 1 \leq h,h' \leq 2^{k_i}}}
\]
(following the notation introduced in Lemma~\ref{lem:coefficients})

$$
=
\begin{cases}
             1 & \text{if $y \oplus f(x) = 0$ and $x = x'$}, \\
            -1 & \text{otherwise}.                 
\end{cases}
$$

Denote by $\sigma_{f} \in \{0,1,-1\}^{\mathcal{P}_{T}}$ the corresponding 
sign condition on $\mathcal{P}_{T}$. Hence, each Boolean function $f:\{0,1\}^n \rightarrow \{0,1\}$ 
computable by a quantum circuit $C$ in $T$ 
defines 
a realizable sign condition $\sigma_f \in \Sign(\mathcal{P}_T)$. Moreover, if $f \neq f'$, then
$\sigma_{f} \neq \sigma_{f'}$. This implies that the number of distinct Boolean functions
$f:\{0,1\}^n \rightarrow \{0,1\}$ 
which are computable by a quantum circuit $C$ in $T$, 
is bounded by $\card(\Sign(\mathcal{P}_{T}))$.

Now, $\card(\mathcal{P}_{T}) = 2^n \cdot 2^{n+1} \leq  2^{2n+1}$.
The degrees of the polynomials in $\mathcal{P}_T$ are bounded by $2 r$, and the number of indeterminates 
by $2 \cdot 2^{2 q} \cdot r$.
We obtain using Proposition~\ref{prop:sign-condition} that 
\[
\card(\Sign(\mathcal{P}_T)) \leq (O(2^{2n+1} \cdot 2\cdot r))^{2 \cdot 2^{2q}\cdot  r} = (2^{n} \cdot r)^{O(2^{2q} \cdot r)}.
\]
\end{proof} 

\begin{remark}
\label{rem:unitarity}
    Note that in the proof of Lemma~\ref{lem:0-1} we are not using the fact the real dimension of the unitary group $\mathbf{U}(2^k)$ is equal to $2^{2k}$. It is possible to take this into account and use a more refined estimate on the number of realizable sign conditions whose combinatorial part (i.e. the part depending on the number of polynomials)  depends on the dimension of the ambient real variety (see \cite{BPR02}). However, this would only improve the constant in our theorem at the cost of introducing more technicalities and so we avoid making this more refined analysis.   
\end{remark}

\begin{proof}[Proof of Theorem~\ref{thm:main}]
Multiplying the bounds in  Lemmas~\ref{lem:topology} and \ref{lem:0-1} 
we obtain that 
that the number of Boolean functions $\{0,1\}^n \rightarrow \{0,1\}$ that can be computed 
 using quantum circuits in $\mathcal{C}_{q,n,r,t} $ 
\[
q^r \cdot (n+t)^{q \cdot r} \cdot (2^{n} \cdot r)^{O(2^{2q} \cdot r)}  =  t^{q \cdot r} \cdot  (2^{n}\cdot  r)^{O(2^{2q} \cdot r)}.
\]
\end{proof}

\begin{proof}[Proof of Corollary~\ref{cor:main}]
From Theorem~\ref{thm:main} there exists $c > 0$ such that the number of distinct Boolean functions
$f: \{0,1\}^n \rightarrow \{0,1\}$ that can be computed  by quantum circuits 
in $\mathcal{C}_{q,n,r,2^n}$ 
is bounded by 
$(2^{n} \cdot r)^{c \cdot r}$ for some constant $ c=c_q$ depending only on $q$.

Suppose that $r \leq \frac{2^n}{4 c n}$.

Then, the number of Boolean functions that  can be computed by such circuits is bounded by 
\[
 (2^{n} \cdot  r)^{c \cdot  r}
 \leq (2^n  \cdot \frac{2^n}{4c n})^{\frac{2^n}{4n}} \leq (2^{2n})^{\frac{2^n}{4n}} = 2^{2^{n-1}}.
\]
Thus, the fraction of all Boolean formulas in $n$ variables that can be computed by such circuits is bounded by 
\[
\frac{2^{2^{n-1}}}{2^{2^{n}}} = 2^{-2^{n-1}} = o(1).
\]
\end{proof}

\section{Alternative approach using approximation and counting}
\label{sec:alternative}
In this section we explore an alternative approach towards proving Corollary~\ref{cor:main}.
It is a classical result due to Solovay and Kitaev \cite{Kitaev-Solovay} that there exists $c >0$, such that for
for any fixed subset $\mathcal{U}' \subset \mathbf{U}_2$ generating a dense subgroup of 
$\mathbf{U}(2)$
and each $\eps > 0$,
an arbitrary unitary matrix $U \in \mathbf{U}(2)$
can be approximated by an element $U' \in \mathbf{U}(2)$ with $||U - U'|| \leq \eps$ 
where $U'$ is a product of at most $\log^c(1/\eps)$ elements of $\mathcal{U}'$ and $||\cdot||$ is the operator norm. 

There has been many later improvements on this fundamental result reducing the value of the 
constant $c$. In particular, it is possible to choose a finite set of gets (Clifford and Toffoli gates)
for which one can take $c=1$ \cite{Maslov-et-al,Selinger}. Using this fact it is not too difficult to prove the following.

\begin{proposition}
\label{prop:approximation}
    For each $\eps >0$, and any quantum circuit $C \in \mathcal{C}){q,n,r,t}$, there exists another circuit
    $C' \in \mathcal{C}_{q,n,r',t}$ which uses only Clifford and Toffoli gates, such that 
    $||U(C) - U(C')|| \leq \eps$ and 
    \[
        r' = O(r \log r + r \log(1/\eps)).
    \]
\end{proposition}

\begin{proof}[Proof sketch]
Replace each gate in $C$ by a circuit using only Clifford and Toffoli gates such that 
the error in norm is bounded by $\eps/r$. Since there are $r$ gates in $C$ the total error will be bounded by $\eps$.   
\end{proof}

An approach towards proving Corollary~\ref{cor:main} would then be as follows. Given 
any $C \in \mathcal{C}_{q,n,r,t}$ computing a Boolean function $f$, let $C'$ be a circuit
using only Clifford and Toffoli gates that approximates $U(C)$ sufficiently well so that 
$C'$ also computes $f$. Using Proposition~\ref{prop:approximation} one would obtain an upper bound
on the size of $C'$ and since the number of circuits with Clifford and Toffoli gates is finite
one can then use a counting argument. However, Definition~\ref{def:quantum-boolean} gives no room for any approximation, as any error could mean that the new circuit $C'$ does not compute $f$.

One can make the Definition~\ref{def:quantum-boolean} more stringent (thus easier for proving lower bounds). For instance, consider the following definition which we call $\delta$-stringent.

\begin{definition}[($\delta$-stringent)-computation of Boolean functions by a quantum circuit]
\label{def:quantum-boolean-delta-stringent}
Let $0 \leq \delta < 1/2$.
A quantum circuit $C$ on $n$ qubits denoted $x_0, \ldots,x_{n-1}$
and $t+1$ ancillary bits 
denoted by $y,z_1,\ldots,z_t$,
 \emph{computes $f$ $\delta$-stringently}, 
if 
for each $x,x', 0 \leq x,x' \leq 2^{n}-1$, $y \in \{0,1\}$, 

\hide{
\begin{eqnarray*}
|a_{j,j',y,0}|^2 &> & 1/2 \mbox{ if $j' = j$, $y = 0$ and $f(j) = 1$ or $j' = j, y =1$ and $f(j) = 0$}, \\
|a_{j,j',y,0}|^2 &< & 1/2 \mbox{ otherwise},
\end{eqnarray*}
}
\begin{eqnarray}
\label{eqn:def:quantum-boolean}
\nonumber
|\la \mathbf{x}' \;y \; \mathbf{0}  \; |\; U(C) \;|\; \mathbf{x} \; 0 \; \mathbf{0} \ra|^2 &>& 1/2+\delta \text{ if $x' = x$ and $y\oplus f(x_0,\ldots,x_{n-1}) = 1$},\\  
&<&  1/2 - \delta \text{ otherwise.}
\end{eqnarray}

\hide{
$a_{j,j',y,z}$ is defined by the equation,
\[
U(C)(|j \ra \otimes |y\ra \otimes 
| z \ra
=
\sum_{0 \leq j' \leq 2^{n}-1, 0 \leq z \leq 2^{t}-1} a_{j,j',y,z} |j'\ra \otimes |y\ra  \otimes |0\ra \otimes |z\ra.
\]
}
($|\mathbf{x} \ra \otimes |y \ra \otimes |\mathbf{0}\ra$ is abbreviated as $|\mathbf{x} \; y\; \mathbf{0}\ra$).
\end{definition}

\begin{remark}
    Definition~\ref{def:quantum-boolean} is the special case of being $0$-stringent.
\end{remark}

One can carry through the program sketched earlier using approximation and counting, if we take
Definition~\ref{def:quantum-boolean-delta-stringent} as our definition for quantum circuit computing Boolean function with $\delta > 0$. In this case one would need to replace $C$ by a circuit $C'$ using Clifford and Toffoli gates which approximates $U(C)$ within $\delta/2$ in max-norm. 
Using Proposition~\ref{prop:approximation} one can take $r' = O(r \log r + r \log(1/\delta))$.
Since the number of such circuits is bounded by $n^{O(r \log r + r \log(1/\delta))}$.
From the inequality
\[
n^{O(r \log r + r \log(1/\delta))} \geq 2^{2^n}
\]
one derives the lower bound 
\[
r \geq \Omega( \min(2^n/(n \log n), 2^n/((\log(1/\delta) \log n))))
\] 
from which one can conclude that almost all Boolean functions need quantum circuits of size
\[
\Omega( \min(2^n/(n \log n), 2^n/((\log(1/\delta) \log n)))).
\]

Notice that the above lower bound is worse than that in Corollary~\ref{cor:main}, and moreover goes to $0$ as
$\delta \rightarrow 0$ and thus does not produce any meaningful lower bound for $\delta =0$ (which is the case in Definition~\ref{def:quantum-boolean}).

\section{Conclusion and future work}
\label{sec:conclusion}
We have introduced a new algebraic technique for proving lower bounds in quantum complexity theory. This 
might have applications in proving lower bounds for other problems in quantum complexity theory. 
One important feature of our method is that it does not need 
unitarity of the gates (see Remark~\ref{rem:unitarity}). This may be relevant for proving results about the complexity class $\mathbf{PostBQP}$ \cite{Aaronson}. We leave this for future work.

\section*{Acknowledgements} We thank Sergey Bravyi, Elena Grigorescu, Dmitri Maslov and Eric Samperton for their comments on a draft version of the paper which helped to improve the paper.

\bibliographystyle{plain}
\bibliography{master}
\end{document}